\RequirePackage{fix-cm}
\documentclass[smallextended]{svjour3}       
\smartqed  
\usepackage{graphicx}
\usepackage{amssymb, amsmath, braket, xcolor}
\usepackage[all]{xy}

\DeclareMathOperator{\tr}{tr}
\DeclareMathOperator{\diag}{diag}

\def \C{{\mathbb C}}

\newtheorem{pro}{Procedure}
\newtheorem{thm}{Theorem}
\newtheorem{ex}{Example}

\begin{document}

\title{A graph theoretical approach to states and unitary operations \thanks{This work  is supported by CSIR (Council  of Scientific and Industrial Research) Grant No.  25(0210)/13/EMR-II, New Delhi,  India.}}


\author{Supriyo Dutta \and Bibhas   Adhikari \and Subhashish  Banerjee}


\institute{Supriyo Dutta \at
              Department of Mathematics, IIT Jodhpur.\\
              \email{dutta.1@iitj.ac.in}           
           \and
           Bibhas   Adhikari \at
              Department of Mathematics, IIT Kharagpur.\\
              \email{bibhas@maths.iitkgp.ernet.in}
           \and
           Subhashish  Banerjee \at
           		Department of Physics, IIT Jodhpur.\\
           		\email{subhashish@iitj.ac.in}
}

\date{Received: date / Accepted: date}

\maketitle

\begin{abstract}
Building upon our previous work, on graphical representation of a quantum state by signless Laplacian matrix, we pose the following question. If a local unitary operation is applied to a quantum state, represented by a signless Laplacian matrix, what would be the corresponding graph and how does one implement local unitary transformations graphically? We answer this question by developing the notion of local unitary equivalent graphs. We illustrate our method by a few, well known, local unitary transformations implemented by single qubit Pauli and Hadamard gates. We also show how graph switching can be used to implement the action of the $C_{NOT}$ gate, resulting in a graphical description of Bell state generation.
\keywords{Signless Laplacian of a combinatorial graph \and Graph switching \and Pauli matrices \and Local unitary operators}
\end{abstract}

\section{Introduction}
\label{intro}
Humans are  fundamentally visual  beings.  Our  most intuitive
  senses are geometric and topological. Thus having a sense for graphs
  comes naturally  to us. On  the other extreme, perhaps  the subtlest
  concept in physics, farthest removed  from our sensory intuition, is
  the  nature of  quantum superposition,  which asserts  that particle
  properties  may  lack a  definiteness  or  realism in  a  profoundly
  fundamental sense.  Thus it is worthwhile asking how we may leverage
  the lucidity  of the  graphs to represent  quantum states.   A graph
  theoretical  approach  to  quantum  mechanics  would  also  help  to
  amalgamate visualization  offered by graphs with  the well developed
  mathematical  machinery  of  graph  theory.  This  motivated  us  to
  provide a  graphical representation of quantum  states in \cite{us},
  where we showed  that the concept of a signless  Laplacian matrix is
  more advantageous,  for graphical representation of  quantum states,
  than the combinatorial Laplacian matrix.

Here we build upon our  previous construction, of representing quantum
states  by  density  matrices  defined  by  using  signless  Laplacian
matrices associated  with weighted  graphs without multiple  edges and
with/without loops.   Our main  intention is to  establish a  proof of
principle,  and study  the graph  theoretic transformation  associated
with  the simplest  non-trivial operation  on $n$  qubits, namely  the
local unitary transformation of $n$ qubits. This would have direct
relevance to the field of quantum information [\cite{EPR} \dots \cite{NC}]
which has been brought to the realms of practical endeavours by mean of
spectacular experimental advances, such as in \cite{Haroche}, \cite{Wineland}.

If $G$ is a weighted graph with/without loops having real edge weights
and  nonnegative  loop weights,  the  density  matrix defined  by  $G$
is  $$\rho(G)=\frac{1}{\tr(L(G))}L(G)$$ where  $L(G)$ is  the signless
Laplacian matrix associated with the graph $G$, see \cite{us}. Some of
the well  known unitary transformations  are implemented by the Pauli
$X=\begin{bmatrix}0 &  1\\ 1 & 0\end{bmatrix}$,  $Y=\begin{bmatrix}0 &
-i\\  i   &  0\end{bmatrix}$  and   $Z=\begin{bmatrix}1  &  0\\   0  &
-1\end{bmatrix}$    operators    or    their    combination    $H    =
\frac{X+Z}{\sqrt{2}}$,  the Hadamard  operator.  These operations  are
also synonymous with the corresponding single-qubit logic gates. Another
well known two-qubit operation is realized by the so called $C_{NOT}$ gate,
which together with suitable combinations of single-qubit gates can be
shown to be adequate for universal quantum computation \cite{Barenco}.

In this paper, we restrict our attention to weighted undirected graphs
of order  $2^n$ for  any natural  number $n.$  We introduce  new graph
theoretic operations or switching methods for weighted graphs by using
quantum logic  gates (for example,  Pauli matrices) on single qubits
which produce  cospectral and  signless Laplacian  cospectral weighted
undirected  graphs. By  (signless Laplacian)  cospectral we  imply two
graphs  with  equal  multi-set  of eigenvalues  of  the  corresponding
(signless Laplacian) adjacency matrices.

Given a graph $G$ of order  $2^n$ we generate new graphs $G^{U_k}$ by
applying      switching     methods      on     $G$      such     that
$\rho(G^{U_k})=U_k\rho(G)U_k^\dagger$ for some unitary matrix $U_k$ of
the form
\begin{equation}\label{LUO} U_k=U^{(1)} \otimes \hdots \otimes U^{(k-1)} \otimes U \otimes U^{(k+1)} \otimes \hdots \otimes U^{(n)}
\end{equation}
where  $U\in\{X,Y,Z, H\}$ and $U^{(j)}=I_2$ the identity matrix of order $2$ when $j\neq k$, and $k=1,2, \hdots, n.$

It is evident that the unitary matrix $U_k$ of order $2^n$ given in (\ref{LUO}) is a local unitary transformation acting on the Hilbert space ${\C^2}^{\otimes n}\equiv \mathcal{H}_1\otimes \mathcal{H}_2\otimes \hdots \otimes \mathcal{H}_n$ where $\mathcal{H}_j=\C^2$ for $j=1,2, \hdots, n.$ Thus, we show that local unitary operations defined in (\ref{LUO}) when applied on density matrix of a $n$-qubit quantum state $\rho(G)$ obtained by signless Laplacian matrix associated with a weighted graph $G$, can be realized by suitable graph switchings of the graph $G$. We call the switching equivalent graphs $G$ and $G^{U_k}$ the local unitary equivalent graphs.

The plan of  the paper is as  follows. In Section 2  we briefly review
the connection between  the concepts of graphs,  quantum states, graph
switching  and  local unitary  operations.  Next,  we introduce  graph
switching methods  which can realize local  unitary operations applied
to a density matrix associated  with graphs. Thus we provide switching
methods  to  generate  cospectral and  signless  Laplacian  cospectral
weighted undirected  graphs.  $CNOT$ gates  are two
qubit gates that, along with the  single qubit Hadamard gate, are used
to generate the two qubit  maximally entangled Bell states. In Section
$5$  we  show how  the  graph  switching  techniques  can be  used  to
implement  the action  of a  $CNOT$ gate,  resulting in  the graphical
description of Bell state generation. We then make our conclusions.

\section{Some Preliminaries}

A graph $G = (V(G), E(G))$ is a combination of vertex set $V(G)$ and edge set $E(G)$. $G$ is weighted if there is a weight function $w: E(G) \rightarrow \mathbb{R}$, $w(i,j) = w_{i,j}$. Adjacency matrix of $G$ is $A(G) = (a_{i,j})$, where,
$$a_{i,j} = \begin{cases} w_{i,j} &\text{if}~ (i,j) \in E(G)\\ 0 &\text{if}~ (i,j) \notin E(G)\end{cases}.$$
Adjacency matrix represents connections between different vertices in a graph. Degree of vertex $i$ is $d_i = \sum_j|a_{i,j}|$, and is indicative of the role played by the vertex $i$ in the graph $G$. Degree matrix of $G$ is $D(G) = diag\{d_i: i = 1, 2, \dots\}$.
Originally, combinatorial Laplacian $\mathcal{L} = D(G) - A(G)$ was used for representing density matrices corresponding to quantum states \cite{sibasish,wu}. 

In \cite{us}, density matrices were constructed using signless and signed Laplacians. Here, our constructions are made using the signless Laplacian. In \cite{Cvetkovic}, it was shown that the signless Laplacian is most convenient for use in the study of the properties of graphs compared to any other matrix associated with a graph (generalized adjacency matrices). In \cite{BeauGioSevWil}, the authors have obtained an ensemble of density matrices as the normalized version of the signless Laplacian matrices which arise from a uniform mixture of unsigned edge-states in a simple graph.

For a weighted undirected graph $G$ without multiple edges, with/without loops having real edge weights and nonnegative loop weights, the density matrix associated with $G=(V(G), E(G))$ is given by
\begin{equation}
\rho(G)=\frac{1}{\tr(L(G))}L(G), \, L(G)=D(G)+A(G)
\end{equation}
The matrix $L(G)$ of order $|V(G)|$ is called the signless Laplacian matrix associated with $G$ \cite{us}. To normalize the state associated with $L(G)$, we divide it by its trace. If G has no loop, then trace(D(G)) = trace(L(G)). But here, trace(D(G)) $\neq$ trace(L(G)) as graphs with weighted loops are considered. Hence, $\rho(G)$ is a Hermitian positive semi-definite matrix of trace $1$ and hence represents the density matrix of a quantum state in $\C^{|V(G)|}.$ The graph $G$ is called a graph representation of $\rho(G)$. As an example consider $\rho_0 = \ket{0}\bra{0} = \begin{bmatrix}1 &  0\\ 0 & 0\end{bmatrix}$ and $\rho_1 = \ket{1}\bra{1} = \begin{bmatrix}0 & 0 \\ 0 & 1\end{bmatrix}$. Graphs corresponding to $\rho_0$ and $\rho_1$ are
$$
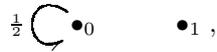
\xymatrix{\bullet_0\ar@(ul,dl)_{\frac{1}{2}} & \bullet_1},$$
and
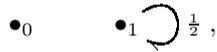
$$\xymatrix{\bullet_0 & \bullet_1\ar@(ur,dr)^{\frac{1}{2}}},$$
respectively. These graphs do not have any edge but have a loop of weight $\frac{1}{2}$.

The following observations can be proved by using arguments similar to those in \cite{us}.
\begin{itemize}
\item The density matrix defined by an undirected weighted graph $G$ without loops of order $n$  represents a pure state if and only if $G$ is isomorphic to $\widehat{K}_2=K_2\sqcup i_1 \sqcup \hdots \sqcup i_{n-2}$, where $K_2$ is the complete graph of order $2$ (a graph with two vertices and one weighted edge).

\item The density matrix defined by an undirected weighted graph $G$ with nonnegative weighted loops of order $n$  represents a pure state if and only if $G$ is isomorphic to $\widehat{O}_1=O_1 \sqcup i_1 \sqcup i_2 \hdots \sqcup i_{n-1}$, where $O_1$ denotes a graph having one node with self-loop.
\end{itemize}

As an example consider the graph in the figure below representing a well known two qubit pure quantum state $\frac{1}{\sqrt{2}}(\ket{01} + \ket{10})$ \cite{us}, the Bell state.
$$\xymatrix{\bullet_0 & \bullet_1\ar@{-}[r]^{1} & \bullet_2 & \bullet_3}.$$
A single qubit mixed state $a\ket{0}\bra{0} + b\ket{1}\bra{1}, a + b = 1$ can be represented by
$$
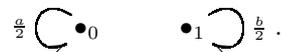
\xymatrix{\bullet_0\ar@(ul,dl)_{\frac{a}{2}} & \bullet_1 \ar@(ur,dr)^{\frac{b}{2}}}.$$

Thus, the density matrix which represents an $n$-qubit multipartite quantum state defined by signless Laplacian matrix associated with an undirected weighted graph $G$ of order $2^n$ with/without loops is a pure state if and only if $G$ contains only one edge. Otherwise, the density matrix $\rho(G)$ represents a mixed state.

Switching of a weighted graph $G$ is a technique to generate a new weighted graph $H$ from $G$ keeping the vertex set fixed. Thus, switching a graph $G=(V(G), E(G))$ means constructing a graph $H=(V(H), E(H))$ such that
\begin{itemize}
\item $V(H)=V(G)$.
\item $E(H)$ is given by $E(G)$ after removing/adding some weighted edges and/or altering weights of the edges in $G.$
\end{itemize}

The graphs $G$ and $H$ are called switching equivalent graphs. The switching method was first proposed in \cite{Seidel} for simple graphs that can produce cospectral simple graphs. Such a switching is well known as Seidel switching in the literature. Recently, some switching methods have been proposed in literature for directed, undirected and weighted graphs to generate cospectral, combinatorial Laplacian cospectral and signless Laplacian cospectral graphs \cite{Butler1}, \cite{Butler2}, respectively.

Familiar single qubit gates are $X, Y, Z$ and $H$. They work on single qubit quantum states $\ket{0}, \ket{1}$ and there linear combinations. Graphically we can represent their actions on $\rho_0$ and $\rho_1$ as:

\begin{enumerate}
\item
{\bf $X$ gate:}
Note that $X\rho_0X^{\dagger} = \rho_1$ and $X\rho_1 X^{\dagger} = \rho_0$. Thus applying $X$ gate on $\rho_0$ is equivalent to removing loop of weight $\frac{1}{2}$ from node $0$ and adding a loop of weight $\frac{1}{2}$ at node $1$. Corresponding action will be observed for $\rho_1$. The graph theoretical representations of these actions are illustrated below.
$$\xymatrix{\bullet_0\ar@(ul,dl)_{\frac{1}{2}} & \bullet_1} \xrightarrow{X} \xymatrix{\bullet_0 & \bullet_1\ar@(ur,dr)^{\frac{1}{2}}},$$
and
$$\xymatrix{\bullet_0 & \bullet_1\ar@(ur,dr)^{\frac{1}{2}}} \xrightarrow{X} \xymatrix{\bullet_0\ar@(ul,dl)_{\frac{1}{2}} & \bullet_1}.$$

\item
{\bf $Y$ gate:}
As $Y\rho_0Y^{\dagger} = \rho_1$ and $Y\rho_1Y^{\dagger} = \rho_0$ and its graphical representation is similar to that of $X$ above.

\item
{\bf $Z$ gate:}
Here $Z\rho_0Z^{\dagger} = \rho_0$ and $Z \rho_1Z^{\dagger} = \rho_1$. Graphically its action are given by
$$\xymatrix{\bullet_0\ar@(ul,dl)_{\frac{1}{2}} & \bullet_1} \xrightarrow{Z} \xymatrix{\bullet_0\ar@(ul,dl)_{\frac{1}{2}} & \bullet_1},$$
and
$$\xymatrix{\bullet_0 & \bullet_1\ar@(ur,dr)^{\frac{1}{2}}} \xrightarrow{Z} \xymatrix{\bullet_0 & \bullet_1\ar@(ur,dr)^{\frac{1}{2}}}.$$

\item
{\bf $H$ gate:}
$H\rho_0H^{\dagger} = \frac{1}{2}\begin{bmatrix} 1 & 1 \\ 1 & 1\end{bmatrix}$, which is the density matrix corresponding to a graph with two vertices. $H\rho_1H^{\dagger} = \frac{1}{2}\begin{bmatrix} 1 & -1 \\ -1 & 1\end{bmatrix}$, corresponding to a graph with two vertices with an edge of weight -$1$. Thus applying $H$ gate on $\rho_0$ or $\rho_1$ is tantamount to removing loops from the nodes and adding an edge between the two nodes. This can be represented graphically as:
$$\xymatrix{\bullet_0\ar@(ul,dl)_{\frac{1}{2}} & \bullet_1} \xrightarrow{H} \xymatrix{\bullet_0 \ar@{-}[r]^1 & \bullet_1},$$
and
$$\xymatrix{\bullet_0 & \bullet_1\ar@(ur,dr)^{\frac{1}{2}}} \xrightarrow{H} \xymatrix{\bullet_0 \ar@{-}[r]^{-1}& \bullet_1}.$$

\item
{\bf $CNOT$ gate:}
$CNOT$ is a two qubit gate whose matrix representation is denoted here by $C_{NOT}$.
$$C_{NOT} = \begin{bmatrix}I_2 & 0 \\ 0 & X \end{bmatrix} = \begin{bmatrix} 1 & 0 & 0 & 0 \\ 0 & 1 & 0 & 0 \\ 0 & 0 & 0 & 1 \\ 0 & 0 & 1 & 0 \end{bmatrix}.$$
$CNOT$ gate is useful for generating Bell states. Its graphical action on two qubit states is depicted at the end of section $5$,
\end{enumerate}

An $n$-qubit quantum state can be represented by a graph of order $2^n$. As an example a $3$-qubit quantum state $\frac{1}{\sqrt{2}}(\ket{001} + \ket{110})$ can be represented as
$$
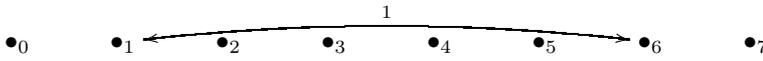
\xymatrix{\bullet_0 & \bullet_1\ar@/^/[rrrrr]^1 & \bullet_2 & \bullet_3 & \bullet_4 & \bullet_5 & \bullet_6\ar@/_/[lllll] & \bullet_7}$$
Another, well known, example of a $3$-qubit quantum state $\ket{W} = \frac{1}{\sqrt{3}}(\ket{100} + \ket{010} + \ket{001})$ can be represented \cite{us} as
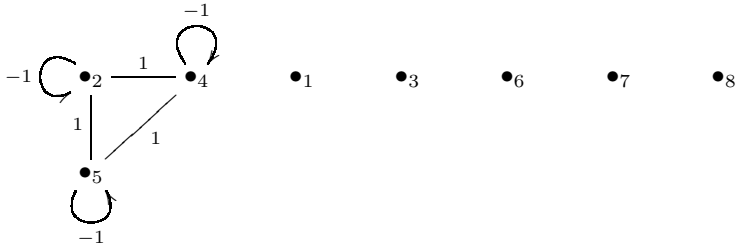
$$\xymatrix{\bullet_2  \ar@{-}[r]^1 \ar@{-}[d]_1 \ar@(ul,dl)_{-1}& \bullet_4 \ar@(lu,ru)^{-1} & \bullet_1 & \bullet_3 & \bullet_6 & \bullet_7 & \bullet_8 \\ \bullet_5 \ar@{-}[ur]_1 \ar@(ld,rd)_{-1} & & & & & &}$$
This graph corresponds to a state represented by a signed Laplacian, but here we do not make further use of this.

Action of local unitary operations on multi-partite systems have found
place  in the  core  of  quantum information  due  to applications  in
several   quantum   information   protocols.  For   example,   quantum
teleportation and quantum dense coding are based on the equivalence of
some classes of states of bi-partite systems \cite{Bennett}. Since the
entries of a density matrix representing a multi-partite quantum state
depend on the choice of the basis in Hilbert space associated with the
system,  action of  a local  unitary operation  implies change  of the
basis in the subsystems. Thus,  a local unitary operation only changes
our point of view without affecting the physical system.

Consider an $n$-qubit system given by the Hilbert space ${\C^2}^{\otimes n}.$ Then two $n$-qubit states given by the density matrices $\rho_1$ and $\rho_2$ are said to be local unitary equivalent if there exists a unitary matrix $U$ of order $2^n$ for which $\rho_2=U\rho_1U^\dagger$ where $U$ is of the form $U=U_1\otimes U_2 \otimes \hdots \otimes U_n, U_i$ is a unitary matrix of order $2$ for $i=1,2, \hdots, n.$

In section 3, we develop switching methods for weighted undirected graphs such that for the two switching equivalent graphs $G$ and $G^{U_k}$, $A(G^{U_k}) = U_kA(G)U_k^\dagger$ and
$$\rho(G^{U_k}) = U_k\rho(G)U_k^\dagger, \, U_k=U^{(1)} \otimes \hdots \otimes U^{(k-1)} \otimes U \otimes U^{(k+1)} \otimes \hdots \otimes U^{(n)},$$ where,  $U\in\{X,Y,Z, H\}$ and $U^{(j)}=I_2$ the identity matrix of order $2$ when $j\neq k$, and $k=1,2, \hdots, n$. Thus, the switching methods proposed in this paper preserve both the spectra and signless Laplacian spectra. Thereby, we introduce the concept of local unitary equivalent graphs which are useful for the realization of local unitary operations on $n$-qubit quantum states, represented by weighted undirected graphs of order $2^n.$

\section{Switching methods and local unitary equivalent $n$-qubit quantum states}

In order to interpret switching equivalent weighted undirected graphs as local unitary equivalent $n$-qubit quantum states which are represented by graphs, we consider a block representation of the adjacency matrix associated with such a graph. Thus, given a weighted undirected graph $G=(V(G), E(G))$ of order $2^n,$ we consider a partition of the vertex set $V(G)=\{0, 1, \hdots, 2^n -1\}$ given by $$V(G)=\sqcup_{j=0}^{2^{n-1} -1} C_j$$ where $C_j=\{2j, 2j+1\}, j=0,1, \hdots,  2^{n-1}-1.$ Then the the adjacency matrix of $G$ is given by
\begin{equation}
\begin{split}
A(G) & = \begin{bmatrix}
C_{0,0} & C_{0,1} & \dots & C_{0,(2^{n-1} - 1)}\\
C_{1,0} & C_{1,1} & \dots & C_{1,(2^{n-1} - 1)}\\
\vdots & \vdots & \vdots & \vdots\\
C_{(2^{n-1} - 1),0} & C_{(2^{n-1} - 1),1} & \dots & C_{(2^{n-1} - 1),(2^{n-1} - 1)}
\end{bmatrix}_{2^n \times 2^n}\\
\text{where} \hspace{1mm} C_{i,j} & = \begin{bmatrix} w(2i,2j)|_G & w(2i,2j+1)|_G \\ w(2i+1,2j)|_G & w(2i+1,2j+1)|_G\end{bmatrix}_{2 \times 2}\\
\text{and} \hspace{1mm} C_{i,i}& = \begin{bmatrix} w(2i,2i)|_G & w(2i, 2i+1)|_G \\ w(2i+1,2i)|_G & w(2i+1,2i+1)|_G \end{bmatrix}_{2 \times 2}.\\
\end{split}
\end{equation}

Obviously, $C_{i,j} = C_{j,i}.$ Note that $C_{i,i}$ is the adjacency matrix corresponding to a subgraph induced by the vertex subset $C_i$ and $C_{i,j}$ provides the information of the existence of edges between $C_i$ and $C_j$ for all $i, j$.

First we propose switching methods which can realize the local unitary operations given by $U_n=\underbrace{I_2\otimes I_2\otimes \hdots \otimes I_2}_{(n-1) \, \mbox{times}}\otimes U$ where $U = \begin{bmatrix} e^{i\phi_1}\cos(\theta) & e^{i\phi_2}\sin(\theta) \\ -e^{-i\phi_2}\sin(\theta) & e^{-i\phi_1}\cos(\theta) \end{bmatrix}.$

\begin{pro}
Given an weighted undirected graph $G=(V(G), E(G))$ we propose the following rule to generate a new graph $G^{U_n}=(V(G), E(G^{U_n}))$. Each step of this procedure will be applied on different edges or loops present in the original graph one by one. To generate the resultant graph if more than one weights are assigned during the process on a single edge or loop, ultimately, all the assigned weights will be added up cumulatively to produce the resultant weight for each edge or loop.
\begin{enumerate}
	\item
	Assign the following edge weights inside every module $C_i$ when $E(G)$ contains loops in it.
		\begin{enumerate}
		\item
			Let $(2i,2i) \in E(G)$, then
			\begin{enumerate}
				\item
					$(2i,2i) \in E(G^{U_n}), \\ w(2i,2i)|_{G^{U_n}} = e^{2i\phi_1}\cos^2(\theta)w(2i,2i)|_G$
				\item
					$(2i,2i + 1) \in E(G^{U_n}), \\ w(2i,2i + 1)|_{G^{U_n}} = e^{i(\phi_1 + \phi_2)}\cos(\theta)\sin(\theta)w(2i,2i)|_G$
				\item
					$(2i+ 1,2i) \in E(G^{U_n}), \\ w(2i + 1,2i)|_{G^{U_n}} = -e^{i(\phi_1 - \phi_2)}\sin(\theta)\cos(\theta)w(2i,2i)|_G$
				\item
					$(2i + 1,2i + 1) \in E(G^{U_n}), \\ w(2i + 1,2i + 1)|_{G^{U_n}} = -\sin^2(\theta)w(2i,2i)|_G$
			\end{enumerate}
		\item
		Let $(2i + 1,2i + 1) \in E(G)$, then
			\begin{enumerate}
				\item
					$(2i,2i) \in E(G^{U_n}), \\ w(2i,2i)|_{G^{U_n}} = -\sin^2(\theta)w(2i + 1,2i + 1)|_G$
				\item
					$(2i,2i + 1) \in E(G^{U_n}), \\ w(2i,2i + 1)|_{G^{U_n}} = e^{-i(\phi_1 - \phi_2)}\cos(\theta)\sin(\theta)w(2i + 1,2i + 1)|_G$
				\item
					$(2i+ 1,2i) \in E(G^{U_n}), \\ w(2i + 1,2i)|_{G^{U_n}} = -e^{-i(\phi_1 + \phi_2)}\sin(\theta)\cos(\theta)w(2i + 1,2i + 1)|_G$
				\item
					$(2i + 1,2i + 1) \in E(G^{U_n}), \\ w(2i + 1,2i + 1)|_{G^{U_n}} = e^{-2i\phi_1}\cos^2(\theta)w(2i + 1,2i + 1)|_G$
			\end{enumerate}
		\end{enumerate}
	\item
	Assign the following edge weights inside a module when there is an edge inside a module.
		\begin{enumerate}
		\item
		Let $(2i,2i + 1) \in E(G)$, then
			\begin{enumerate}
				\item
					$(2i,2i) \in E(G^{U_n}), \\ w(2i,2i)|_{G^{U_n}} = -e^{i(\phi_1 - \phi_2)}\cos(\theta)\sin(\theta)w(2i,2i+1)|_G$
				\item
					$(2i,2i + 1) \in E(G^{U_n}), \\ w(2i,2i + 1)|_{G^{U_n}} = \cos^2(\theta)w(2i,2i+1)|_G$
				\item
					$(2i+ 1,2i) \in E(G^{U_n}), \\ w(2i + 1,2i)|_{G^{U_n}} = -e^{-2i\phi_2}\sin^2(\theta)w(2i,2i+1)|_G$
				\item
					$(2i + 1,2i + 1) \in E(G^{U_n}), \\ w(2i + 1,2i + 1)|_{G^{U_n}} = -e^{-i(\phi_1 + \phi_2)}\sin(\theta)	\cos(\theta)w(2i,2i+1)|_G$
			\end{enumerate}
		\item
		Let $(2i + 1,2i) \in E(G)$, then
			\begin{enumerate}
				\item
					$(2i,2i) \in E(G^{U_n}), \\ w(2i,2i)|_{G^{U_n}} = e^{i(\phi_1 + \phi_2)}\cos(\theta)\sin(\theta)w(2i+1,2i)|_G$
				\item
					$(2i,2i + 1) \in E(G^{U_n}), \\ w(2i,2i + 1)|_{G^{U_n}} = e^{2i\phi_2}\sin^2(\theta)w(2i+1,2i)|_G$
				\item
					$(2i+ 1,2i) \in E(G^{U_n}), \\ w(2i + 1,2i)|_{G^{U_n}} = \cos^2(\theta)w(2i+1,2i)|_G$
				\item
					$(2i + 1,2i + 1) \in E(G^{U_n}), \\ w(2i + 1,2i + 1)|_{G^{U_n}} = e^{-i(\phi_1 - \phi_2)}\sin(\theta)\cos(\theta)w(2i+1,2i)|_G$
			\end{enumerate}
		\end{enumerate}
\item
Assign these edge weights when there are edges joining vertices of different modules.
\begin{enumerate}
\item
			Let $(2i,2j) \in E(G)$, then
			\begin{enumerate}
				\item
					$(2i,2j) \in E(G^{U_n}), \\ w(2i,2j)|_{G^{U_n}} = e^{2i\phi_1}\cos^2(\theta)w(2i,2j)|_G$
				\item
					$(2i,2j + 1) \in E(G^{U_n}), \\ w(2i,2j + 1)|_{G^{U_n}} = e^{i(\phi_1 + \phi_2)}\cos(\theta)\sin(\theta)w(2i,2j)|_G$
				\item
					$(2i+ 1,2j) \in E(G^{U_n}), \\ w(2i + 1,2j)|_{G^{U_n}} = -e^{i(\phi_1 - \phi_2)}\sin(\theta)\cos(\theta)w(2i,2j)|_G$
				\item
					$(2i + 1,2j + 1) \in E(G^{U_n}), \\ w(2i + 1,2j + 1)|_{G^{U_n}} = -\sin^2(\theta)w(2i,2j)|_G$
			\end{enumerate}
\item
		Let $(2i + 1,2j + 1) \in E(G)$, then
			\begin{enumerate}
				\item
					$(2i,2j) \in E(G^{U_n}), \\ w(2i,2j)|_{G^{U_n}} = -\sin^2(\theta)w(2i + 1,2j + 1)|_G$
				\item
					$(2i,2j + 1) \in E(G^{U_n}), \\ w(2i,2j + 1)|_{G^{U_n}} = e^{-i(\phi_1 - \phi_2)}\cos(\theta)\sin(\theta)w(2i + 1,2j + 1)|_G$
				\item
					$(2i+ 1,2j) \in E(G^{U_n}), \\ w(2i + 1,2j)|_{G^{U_n}} = -e^{-i(\phi_1 + \phi_2)}\sin(\theta)\cos(\theta)w(2i + 1,2j + 1)|_G$
				\item
					$(2i + 1,2j + 1) \in E(G^{U_n}), \\ w(2i + 1,2j + 1)|_{G^{U_n}} = e^{-2i\phi_1}\cos^2(\theta)w(2i + 1,2j + 1)|_G$	
			\end{enumerate}
\item
		Let $(2i,2j + 1) \in E(G)$, then
			\begin{enumerate}
				\item
					$(2i,2j) \in E(G^{U_n}), \\ w(2i,2j)|_{G^{U_n}} = -e^{i(\phi_1-\phi_2)}\cos(\theta)\sin(\theta)w(2i,2j + 1)|_G$
				\item
					$(2i,2j + 1) \in E(G^{U_n}), \\ w(2i,2j + 1)|_{G^{U_n}} = \cos^2(\theta)w(2i,2j + 1)|_G$
				\item
					$(2i+ 1,2j) \in E(G^{U_n}), \\ w(2i + 1,2j)|_{G^{U_n}} = -e^{-2i\phi_2}\sin^2(\theta)w(2i,2j + 1)|_G$
				\item
					$(2i + 1,2j + 1) \in E(G^{U_n}), \\ w(2i + 1,2j + 1)|_{G^{U_n}} = -e^{-i(\phi_1 + \phi_2)}\sin(\theta)\cos(\theta)w(2i,2j + 1)|_G$
			\end{enumerate}
\item
		Let $(2i + 1,2j) \in E(G)$, then
			\begin{enumerate}
				\item
					$(2i,2j) \in E(G^{U_n}), \\ w(2i,2j)|_{G^{U_n}} = e^{i(\phi_1 + \phi_2)}\cos(\theta)\sin(\theta)w(2i + 1,2j)|_G$
				\item
					$(2i,2j + 1) \in E(G^{U_n}), \\ w(2i,2j + 1)|_{G^{U_n}} = e^{2i\phi_2}\sin^2(\theta)w(2i + 1,2j)|_G$
				\item
					$(2i+ 1,2j) \in E(G^{U_n}), \\ w(2i + 1,2j)|_{G^{U_n}} = \cos^2(\theta)w(2i + 1,2j)|_G$
				\item
					$(2i + 1,2j + 1) \in E(G^{U_n}),\\ w(2i + 1,2j + 1)|_{G^{U_n}} = e^{-i(\phi_1 - \phi_2)}\sin(\theta)\cos(\theta)w(2i + 1,2j)|_G$
			\end{enumerate}
\end{enumerate}
\end{enumerate}
\end{pro}

Then we have the following theorem.

\begin{thm}\label{proofz}
Let $G$ be a weighted undirected graph of order $2^n.$ Then $\rho(G^{U_n}) = U_n \rho(G) U_n^{\dagger}$ where $U_n=\underbrace{I_2\otimes I_2\otimes \hdots \otimes I_2}_{(n-1) \, \mbox{times}}\otimes U.$
\end{thm}

\begin{proof} Note that,
\begin{align*}
{U_n} & = \diag\{U, U, \dots U\}_{2^n \times 2^n} \,\, \mbox{and}\\
U_nA(G)U_n^{\dagger} & = \begin{bmatrix}
UC_0U^{\dagger} & UC_{0,1}U^{\dagger} & \dots & UC_{0,(2^{n-1} - 1)}U^{\dagger}\\
UC_{1,0}U^{\dagger} & UC_1U^{\dagger} & \dots & UC_{1,(2^{n-1} - 1)}U^{\dagger}\\
\vdots & \vdots & \vdots & \vdots\\
UC_{(2^{n-1} - 1),0}U^{\dagger} & UC_{(2^{n-1} - 1),1}U^{\dagger} & \dots & UC_{(2^{n-1} - 1)}U^{\dagger}
\end{bmatrix}.\\
\end{align*}
Here, $C_i$ and $C_{i,j}$ are either the matrices given below or their sum.
$$\begin{bmatrix} w & 0 \\ 0 & 0\end{bmatrix}, \begin{bmatrix} 0 & 0 \\ 0 & w\end{bmatrix}, \begin{bmatrix} 0 & w \\ 0 & 0\end{bmatrix}, \begin{bmatrix} 0 & 0 \\ w & 0\end{bmatrix}.$$
Remaining proof follows from the fact that
\begin{align*}
U\begin{bmatrix} w & 0 \\ 0 & 0\end{bmatrix}U^{\dagger} & =
\begin{bmatrix}
e^{2 i.\phi_1} w \cos[\theta]^2 & e^{i.\phi_1+i.\phi_2} w \cos[\theta] \sin[\theta ] \\
-e^{i.\phi_1-i.\phi_2} w \cos[\theta] \sin[\theta] & -w \sin[\theta ]^2 \\
\end{bmatrix},\\
U\begin{bmatrix} 0 & 0 \\ 0 & w\end{bmatrix}U^{\dagger} & =
\begin{bmatrix}
-w \sin[\theta ]^2 & e^{-i.\phi_1+i.\phi_2} w \cos[\theta ] \sin[\theta ] \\
-e^{-i.\phi_1-i.\phi_2} w \cos[\theta ] \sin[\theta ] & e^{-2 i.\phi_1} w \cos[\theta ]^2
\end{bmatrix},\\
U\begin{bmatrix} 0 & w \\ 0 & 0\end{bmatrix}U^{\dagger} & =
\begin{bmatrix}
-e^{i.\phi_1-i.\phi_2} w \cos[\theta] \sin[\theta ] & w \cos[\theta ]^2 \\
e^{-2 i.\phi_2} w \sin[\theta ]^2 & -e^{-i.\phi_1-i.\phi_2} w \cos[\theta ]\sin[\theta ]
\end{bmatrix},\\
U\begin{bmatrix} 0 & 0 \\ w & 0\end{bmatrix}U^{\dagger} & =
\begin{bmatrix}
e^{i.\phi_1+i.\phi_2} w \cos[\theta] \sin[\theta] & e^{2 i.\phi_2} w \sin[\theta]^2 \\
w \cos[\theta]^2 & e^{-i.\phi_1+i.\phi_2} w \cos[\theta] \sin[\theta]
 \end{bmatrix}.
\end{align*}
Hence proved.
\end{proof}

Here $U$ is the most general unitary matrix of order $2$, as given above. The well known single qubit matrices can be obtained from it as special cases. For $\theta = \frac{\pi}{2}, \phi_2 = \frac{\pi}{2}, U = iX$; for $\theta = \frac{\pi}{2}, \phi_2 = \frac{3\pi}{2}, U = Y$; for $\theta = 0, \phi_1 = \frac{\pi}{2}, U = iZ$ and for $\theta = \frac{\pi}{4}, \phi_1 = \frac{\pi}{2}, \phi_2 = \frac{\pi}{2}, U = iH$. The complex number $i$ acts as a phase factor having no role in unitary operations. Thus, for example, given a weighted undirected graph $G=(V(G), E(G))$, a new graph $G^{X_n}=(V(G), E(G^{X_n}))$, corresponding to $X_n=\underbrace{I_2\otimes I_2\otimes \hdots \otimes I_2}_{(n-1) \, \mbox{times}}\otimes X$, would be generated according to the following rules.

\begin{enumerate}

\item
Interchange loop weights inside every module $C_i$
\begin{equation} \label{x1}
\begin{split}
(2i, 2i) \in E(G) \Rightarrow & (2i + 1, 2i + 1) \in E(G^{X_n}), (2i,2i) \notin E(G^{X_n})\\
& w(2i + 1, 2i + 1)|_{G^{X_n}} = w(2i, 2i)|_G,\\
(2i + 1, 2i + 1) \in E(G) \Rightarrow & (2i, 2i) \in E(G^{X_n}), (2i + 1, 2i + 1) \notin E(G^{X_n})\\
& w(2i, 2i)|_G = w(2i + 1, 2i + 1)|_{G^{X_n}}.
\end{split}
\end{equation}

\item
Keep edge weights inside modules unchanged
\begin{equation}\label{x2}
\begin{split}
(2i,2i+1) \in E(G) \Rightarrow & (2i,2i+1) \in E(G^{X_n}) \\
& w(2i, 2i + 1)|_{G^{X_n}} = w(2i, 2i + 1)|_G.
\end{split}
\end{equation}

\item
Modifications of edges joining two modules
\begin{equation}\label{x3}
\begin{split}
(2i , 2j) \in E(G) \Rightarrow & (2i + 1, 2j + 1) \in E(G^{X_n}), (2i , 2j) \notin E(G^{X_n}) \\
& w(2i + 1, 2j + 1)|_{G^{X_n}} =  w(2i, 2j)|_G,\\
(2i + 1 , 2j + 1) \in E(G) \Rightarrow & (2i, 2j) \in E(G^{X_n}), (2i + 1, 2j + 1) \notin E(G^{X_n}) \\
& w(2i, 2j)|_{G^{X_n}} = w(2i + 1, 2j + 1)|_G,\\
(2i , 2j + 1) \in E(G) \Rightarrow & (2j, 2i + 1) \in E(G^{X_n}), (2i , 2j + 1) \notin E(G^{X_n}) \\
& w(2j, 2i + 1)|_{G^{X_n}} = w(2i , 2j + 1)|_G,\\
(2j , 2i + 1) \in E(G) \Rightarrow & (2i, 2j + 1) \in E(G^{X_n}), (2j , 2i + 1) \notin E(G^{X_n})\\
& w(2i, 2j + 1)|_{G^{X_n}} = w(2j , 2i + 1)|_G.
\end{split}
\end{equation}
\end{enumerate}

The following example provides an overview of the effect of local unitary operators $X_n, Y_n, Z_n, H_n$ when applied to a $2$-qubit density matrix of order $4$, making use of the graph representations of the corresponding density matrix.

\begin{ex}
Consider the graph,
$$
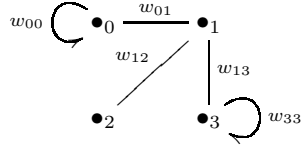
\xymatrix{{\bullet}_{0} \ar@{-}[r]^{w_{01}} \ar@(ul,dl)_{w_{00}} & {\bullet}_{1} \ar@{-}[d]^{w_{13}} \ar@{-}[ld]_{w_{12}} \\ {\bullet}_{2} & {\bullet}_{3} \ar@(ur,dr)^{w_{33}}}.$$
It represents following mixed two qubit state,
$$\rho = \frac{L(G)}{tr(L(G))} = \frac{1}{tr(L(G))}\begin{bmatrix} 2w_{00} + w_{01} & w_{01} & 0 & 0\\ w_{01} & w_{01} + w_{12} + w_{13} & w_{12} & w_{13} \\ 0 & w_{12} & w_{12} & 0 \\ 0 & w_{13} & 0 & w_{13} + 2w_{33}\end{bmatrix},$$ where $tr(L(G)) = 2(w_{00} + w_{01} + w_{12} + w_{13} + w_{33})$.
Applying $X_2$ on the original graph:
$$\xymatrix{{\bullet}_{0} \ar@{-}[r]^{w_{01}} \ar@{-}[rd]^{w_{12}} \ar@{-}[d]_{w_{13}} & {\bullet}_{1} \ar@(ur,dr)^{w_{00}} \\ {\bullet}_{2} \ar@(ul,dl)_{w_{33}} & {\bullet}_{3}}$$
Applying $Y_2$ on the original graph:
$$\xymatrix{{\bullet}_{0} \ar@{-}[r]^{-w_{01}} \ar@{-}[rd]^{-w_{12}} \ar@{-}[d]_{w_{13}}& {\bullet}_{1} \ar@(ur,dr)^{w_{00}} \\ {\bullet}_{2} \ar@(ul,dl)_{w_{33}} & {\bullet}_{3}}$$
Applying $Z_2$ on the original graph:
$$\xymatrix{{\bullet}_{0} \ar@{-}[r]^{-w_{01}} \ar@(ul,dl)_{w_{00}} & {\bullet}_{1} \ar@{-}[d]^{w_{13}} \ar@{-}[ld]_{-w_{12}}\\ {\bullet}_{2} & {\bullet}_{3} \ar@(ur,dr)^{w_{33}}}$$
Applying $H_2$ on the original graph:
$$\xymatrix{{\bullet}_{0} \ar@{-}[r] \ar@{-}[d] \ar@{-}[rd] \ar@(ul,dl)& {\bullet}_{1} \ar@(ur,dr) \ar@{-}[d] \ar@{-}[ld] \\ {\bullet}_{2} \ar@(ul,dl) \ar@{-}[r] & {\bullet}_{3}\ar@(ur,dr)}$$
In this figure, the edge weights $w(i,j)$ where $(i,j) = 0, 1, 2, 3$ are
\begin{equation*}
\begin{split}
w(0,1)&= w_{00} - \frac{1}{2}(w_{12} + w_{13})\\
w(0,2) & = \frac{1}{2}w_{12} + \frac{1}{2}w_{13},\\
w(0,3) & = \frac{1}{2}w_{12} - \frac{1}{2}w_{13},\\
w(1,2) & = -\frac{1}{2}w_{12} - \frac{1}{2}w_{13},\\
w(1,3) & = -\frac{1}{2}w_{12} + \frac{1}{2}w_{13},\\
w(2,3) & = \frac{1}{2}w_{12} - \frac{1}{2}w_{13} - w_{33},\\
w(0,0) & = w_{01} + \frac{w_{13}}{2},\\
w(1,1) & = w_{12} + \frac{w_{13}}{2},\\
w(2,2) & = w_{33} + \frac{w_{13}}{2},\\
w(3,3) & = w_{33} + \frac{w_{13}}{2}.
\end{split}
\end{equation*}
\end{ex}

Thus the local unitary operators given by $X_n, Y_n, Z_n$ and $H_n$ acting on a density matrix $\rho(G)$ of an $n$-qubit quantum state defined by the signless Laplacian associated with a weighted undirected $G$ can be realized by graph switchings applied on the graph $G.$ Now we focus on the local unitary operators $U_k=I_2\otimes \hdots \otimes I_{2}\otimes U\otimes I_2 \otimes \hdots \otimes I_2$ when $k< n$ and $U\in\{X, Y, Z, H\}$ is placed in the $k$-th position of the tensor product.

In order to give a graph theoretic interpretation of $U_k$ we represent the vertex set of a weighted undirected graph $G=(V(G), E(G))$ of order $2^n$ by $V(G)=\{0,1\}^n\equiv \{0,1, \hdots, 2^n-1\}$ where a vertex $j \in V(G)$ is represented by a sequence of $0,1$s. The labeling of the vertices of $G$ is determined by the lexicographic ordering defined on $\{0, 1\}^n.$ For example, if $n=2$ the labeled vertex set is given by $V(G)=\{00, 01, 10, 11\}\equiv \{0, 1, 2, 3\}.$

We also consider a permutation $p_{k,n}$ on the vertex set $V(G)$ defined as follows \begin{eqnarray} p_{k,n} &:& \{0,1\}^n \rightarrow \{0,1\}^n \\ && x_1x_2\hdots x_{k-1}x_kx_{k+1}\hdots x_n \mapsto x_1x_2\hdots x_{k-1}x_nx_{k+1}\hdots x_k\end{eqnarray} where $x_i\in \{0, 1\}.$ Thus given the standard lexicographic ordering on $\{0,1\}^n$, $p_{k,n}$ introduces a relabeling of the vertices. Let $P_{k,n}$ be the unique permutation matrix associated with $p_{k,n}.$ Then it is easy to verify that \begin{equation}\label{uk1} A(G_{p_{k,n}})=P_{k,n}A(G)P_{k,n}^\dagger, \,\,\, D(G_{p_{k,n}})=P_{k,n}D(G)P_{k,n}^\dagger \end{equation} where $G_{p_{k,n}}$ denotes the graph $G$ with a new labeling of the vertices given by $p_{k,n}.$ Moreover, $U_k=P_{k,n}U_nP_{k,n}.$

This yields \begin{equation}\label{uk2}U_k\rho(G)U_k^\dagger = (P_{k,n}U_nP_{k,n})\rho(G)(P_{k,n}U_nP_{k,n})^\dagger =P_{k,n} (U_n(P_{k,n}\rho(G)P_{k,n}^\dagger)U_n^\dagger)P_{k,n}^\dagger.\end{equation}

Therefore, the local unitary operation $U_k$ on an $n$-qubit density matrix $\rho(G)$ defined by a graph $G$ of order $2^n$ can be explained by the following switching procedure \begin{equation}\label{uk} G \mapsto^{p_{k,n}} G_{p_{k,n}} \mapsto^{U_n} G_{p_{k,n}}^{U_n} \mapsto^{p_{k,n}} G^{U_k}\end{equation} where $U_n=I_2\otimes \hdots \otimes I_2 \hdots \otimes U$ and $U\in \{X, Y, Z, H\}.$ Then we have the following theorem.

\begin{thm}\label{proofh}
Let $G$ be a weighted undirected graph of order $2^n.$ Then $\rho(G^{U_k}) = U_k \rho(G) U_k^{\dagger}$ where $U_k=I_1\otimes \hdots \otimes I_{2}\otimes U\otimes I_2 \otimes \hdots \otimes I_2,$ $k< n$ and $U\in\{X, Y, Z, H\}$ placed in the $k$-th position of the tensor product.
\end{thm}
\begin{proof}
From (\ref{uk1}), (\ref{uk2}) and (\ref{uk}) we get the following equation:
\begin{align*}
U_k \rho(G) U_k^{\dagger} & = U_k \frac{1}{\tr(L(G))}L(G) U_k^{\dagger}\\
& = \frac{1}{\tr(L(G))}U_k [D(G) + A(G)]U_k^{\dagger}\\
& = \frac{1}{\tr(L(G))}P_{k,n}U_nP_{k,n}[D(G) + A(G)](P_{k,n}U_nP_{k,n})^{\dagger}\\
& = \frac{1}{\tr(L(G))}[D(G^{U_k}) + A(G^{U_k})].\\
\end{align*}
$P_{k,n}U_n$ is a unitary matrix. Thus $\tr(L(G)) = \tr(L(G^k))$. Replacing it in the above equation we get,
$$U_k \rho(G) U_k^{\dagger} = \frac{1}{\tr(L(G^k))}[D(G^{U_k}) + A(G^{U_k})] = \frac{1}{\tr(L(G^k))}L(G^k) = \rho(G^{U_k}).$$
Hence proved.
\end{proof}

At the end of this section we would like to propose a simplified version of the graph theoretic interpretation of the operations $X_n, Y_n$ and $Z_n$, discussed above. First, partition $V(G)=B\sqcup R$ into halves with colored vertices, say blue and red vertices, where the blue colored vertices are in $B=\{b_1, b_2, \dots b_{2^{n-1}}\}$ and the red colored vertices are in $R=\{r_1, r_2, \dots r_{2^{n-1}}\}.$ Define two vertices of different colors say $b_i$ and $r_j$ as conjugate to each other if $i = j$. A loop $(b_i,b_i)$ or $(r_i,r_i)$ is an edge joining vertices of same color, that is, $b_i$ and $b_i$ or $r_i$ and $r_i$. To illustrate all the transformations discussed above, we consider an example of a graph $G$ of order $4$, giving a $2$-qubit state with edge weights $a, b, c, d$ and $e$, as shown in the figure below.
$$
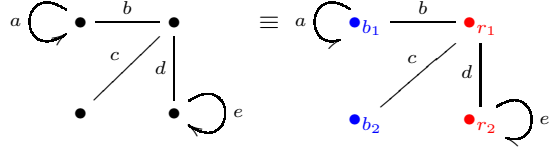
\xymatrix{\bullet \ar@(ul,dl)_{a} \ar@{-}[r]^{b}& \bullet \ar@{-}[dl]_{c} \ar@{-}[d]_{d} \\ \bullet & \bullet \ar@(ur,dr)^{e}} \equiv \xymatrix{\textcolor{blue}{\bullet_{b_1}} \ar@(ul,dl)_{a} \ar@{-}[r]^{b}& \textcolor{red}{\bullet_{r_1}} \ar@{-}[dl]_{c} \ar@{-}[d]_{d} \\ \textcolor{blue}{\bullet_{b_2}} & \textcolor{red}{\bullet_{r_2}} \ar@(ur,dr)^{e}}$$

Following changes in $E(G)$ switch $G$ into $G^{X_n}$:
\begin{enumerate}
\item
Given an edge joining vertices of same color (suppose $r_i$ and $r_j$) in $G$, $E(G^{X_n})$ has a member joining corresponding conjugate vertices ($b_i$ and $b_j$) with same edge weight.
\item
Given an edge joining vertices of different colors (say $(b_i, r_j)$ in $G$), $E(G^{X_n})$ has a member joining corresponding conjugate vertices $(b_j, r_i)$ with same edge weight.
\end{enumerate}
No change is required for edges $(b_i, r_i)$. For the particular example of the graph shown above, $G^{X_n}$ is
$$
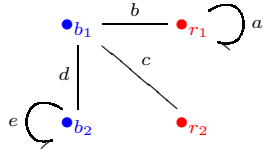
\xymatrix{\textcolor{blue}{\bullet_{b_1}} \ar@{-}[dr]^{c} \ar@{-}[d]_{d} \ar@{-}[r]^{b}& \textcolor{red}{\bullet_{r_1}} \ar@(ur,dr)^{a} \\ \textcolor{blue}{\bullet_{b_2}} \ar@(ul,dl)_{e} & \textcolor{red}{\bullet_{r_2}}}$$

Switching $G$ to $G^{Y_n}$ is equivalent to switching $G$ to $G^{X_n}$ with an additional operation:
\begin{enumerate}
\item
If there is an edge joining vertices of different colors in $G$, the new edge weight is equal to $-1$ times the old edge weight.
\end{enumerate}
Thus, $G^{Y_n}$ is
$$
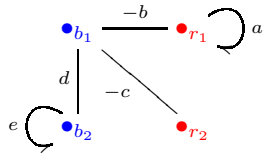
\xymatrix{\textcolor{blue}{\bullet_{b_1}} \ar@{-}[dr]_{-c} \ar@{-}[d]_{d} \ar@{-}[r]^{-b}& \textcolor{red}{\bullet_{r_1}} \ar@(ur,dr)^{a} \\ \textcolor{blue}{\bullet_{b_2}} \ar@(ul,dl)_{e} & \textcolor{red}{\bullet_{r_2}}}$$

Switching $G$ to $G^{Z_n}$ is simplest as it does not change the location of any edge except the edge weight joining vertices of different colors. The new weight is $-1$ times the old edge weight. Hence, $G^{Z_n}$ for the graph $G$, introduced above, is
$$
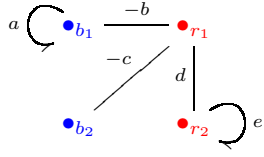
\xymatrix{\textcolor{blue}{\bullet_{b_1}} \ar@(ul,dl)_{a} \ar@{-}[r]^{-b}& \textcolor{red}{\bullet_{r_1}} \ar@{-}[dl]_{-c} \ar@{-}[d]_{d} \\ \textcolor{blue}{\bullet_{b_2}} & \textcolor{red}{\bullet_{r_2}} \ar@(ur,dr)^{e}}$$

It is easy to check that $\rho(G^{U_n}) = U_n\rho(G)U_n^{\dagger}$ for $U = X, Y$ and $Z$ in the above example.

\section{CNOT gate and Bell state generation}
$CNOT$ gates occupy a central position in various quantum information processing tasks. They along with the Hadamard gate can be used for generating the two qubit maximally entangled Bell states from separable states. Here we adapt the switching techniques developed above to the $CNOT$ gate, thereby providing a graphical method of Bell states generation.

Consider super-modules on $V(G)$ as a partition $\{\mathcal{C}_i: i = 0, 1, \dots 2^{n-2}-1\}$, where $$\mathcal{C}_i = C_{2i} \cup C_{2i + 1} = \{4i, 4i + 1\} \cup\{4i + 2, 4i + 3\} = \{4i, 4i + 1, 4i + 2, 4i + 3\}.$$
From a graph $G$ with $|V(G)| = 2^n, n \geq 2$ we construct a graph $G^{C_{NOT}}(V(G^{C_{NOT}}), E^(G^{C_{NOT}}))$, s.t. $\rho(G^{C_{NOT}}) = C_{NOT_n}\rho(G)C_{NOT_n}$, where, $C_{NOT_n} = I \otimes I \otimes \dots \otimes C_{NOT}$. Clearly $V(G^{C_{NOT}}) = V(G)$.

\begin{pro}
{\bf Construct $E(G^{C_{NOT}})$ from $E(G)$ by graph switching}
\begin{enumerate}

\item
Changes inside a supermodule $\mathcal{C}_i: i = 0,1,\dots(2^{n-2}-1)$ shall be as follows.
\begin{enumerate}

\item
No changes in the loops and edges inside module $C_{2i}$
\begin{equation}
\begin{split}
(4i,4i) \in E(G) \Rightarrow & (4i,4i) \in E(G^{C_{NOT}})\\
& w(4i,4i)|_{G^{C_{NOT}}} = w(4i,4i)|_{G},\\
(4i + 1,4i + 1) \in E(G) \Rightarrow & (4i + 1,4i + 1) \in E(G^{C_{NOT}})\\
& w(4i + 1,4i + 1)|_{G^{C_{NOT}}} = w(4i + 1,4i + 1)|_{G},\\
(4i,4i + 1) \in E(G) \Rightarrow & (4i,4i + 1) \in E(G^{C_{NOT}})\\
& w(4i,4i + 1)|_{G^{C_{NOT}}} = w(4i,4i + 1)|_{G}.
\end{split}
\end{equation}

\item
Do the following changes for the edges inside module $C_{2i + 1}$.
\begin{equation}
\begin{split}
(4i + 2,4i + 2) \in E(G) \Rightarrow & (4i + 3,4i + 3) \in E(G^{C_{NOT}}), (4i + 2,4i + 2) \notin E(G^{C_{NOT}})\\ & w(4i + 3,4i + 3)|_{G^{C_{NOT}}} = w(4i + 2,4i + 2)|_{G},\\
(4i + 3,4i + 3) \in E(G) \Rightarrow & (4i + 2,4i + 2) \in E(G^{C_{NOT}}), (4i + 3,4i + 3) \notin E(G^{C_{NOT}})\\ & w(4i + 2,4i + 2)|_{G^{C_{NOT}}} = w(4i + 3,4i + 3)|_{G},\\
(4i + 2,4i + 3) \in E(G) \Rightarrow & (4i + 2,4i + 3) \in E(G^{C_{NOT}})\\
& w(4i + 2,4i + 3)|_{G^{C_{NOT}}} = w(4i + 2,4i + 3)|_{G}.\\
\end{split}
\end{equation}

\item
Change the edges joining modules $C_{2i}$ and $C_{2i + 1}$ of supermodule $\mathcal{C}_i$.
\begin{equation}
\begin{split}
(4i,4i + 2) \in E(G) \Rightarrow & (4i,4i + 3) \in E(G^{C_{NOT}}), (4i,4i + 2) \notin E(G^{C_{NOT}})\\
& w(4i,4i + 3)|_{G^{C_{NOT}}} = w(4i,4i + 2)|_{G},\\
(4i,4i + 3) \in E(G) \Rightarrow & (4i,4i + 2) \in E(G^{C_{NOT}}), (4i,4i + 3) \notin E(G^{C_{NOT}})\\
& w(4i,4i + 2)|_{G^{C_{NOT}}} = w(4i,4i + 3)|_{G},\\
(4i + 1,4i + 2) \in E(G) \Rightarrow & (4i + 1,4i + 3) \in E(G^{C_{NOT}}), (4i + 1,4i + 2) \notin E(G^{C_{NOT}})\\ & w(4i + 1,4i + 3)|_{G^{C_{NOT}}} = w(4i + 1,4i + 2)|_{G},\\
(4i + 1,4i + 3) \in E(G) \Rightarrow & (4i + 1,4i + 2) \in E(G^{C_{NOT}}), (4i + 1,4i + 3) \notin E(G^{C_{NOT}})\\ & w(4i + 1,4i + 2)|_{G^{C_{NOT}}} = w(4i + 1,4i + 3)|_{G}.\\
\end{split}
\end{equation}
\end{enumerate}

\item
Changes in the edges joining different supermodules $\mathcal{C}_i$ and $\mathcal{C}_j$ for $i,j = 0, 1,\dots, (2^{n-2}-1); i\neq j$ shall be as follows.
\begin{enumerate}

\item
No changes in the edges joining $C_{2i} \subset \mathcal{C}_i$ and $C_{2j} \subset \mathcal{C}_j$.
\begin{equation}
\begin{split}
(4i,4j) \in E(G) \Rightarrow & (4i,4j) \in E(G^{C_{NOT}})\\
& w(4i,4j)|_{G^{C_{NOT}}} = w(4i,4j)|_{G},\\
(4i + 1,4j + 1) \in E(G) \Rightarrow & (4i + 1,4j + 1) \in E(G^{C_{NOT}})\\
& w(4i + 1,4j + 1)|_{G^{C_{NOT}}} = w(4i + 1,4j + 1)|_{G},\\
(4i,4j + 1) \in E(G) \Rightarrow & (4i,4j + 1) \in E(G^{C_{NOT}})\\
& w(4i,4j + 1)|_{G^{C_{NOT}}} = w(4i,4j + 1)|_{G},\\
(4i + 1,4j) \in E(G) \Rightarrow & (4i + 1,4j) \in E(G^{C_{NOT}})\\
& w(4i + 1,4j)|_{G^{C_{NOT}}} = w(4i + 1,4j)|_{G}.\\
\end{split}
\end{equation}

\item
Change the edges joining $C_{2i} \subset \mathcal{C}_i$ and $C_{2j + 1} \subset \mathcal{C}_j$ according to:
\begin{equation}
\begin{split}
(4i,4j + 2) \in E(G) \Rightarrow & (4i,4j + 3) \in E(G^{C_{NOT}}), (4i,4j + 2) \notin E(G^{C_{NOT}})\\
& w(4i,4j + 3)|_{G^{C_{NOT}}} = w(4i,4j + 2)|_{G},\\
(4i,4j + 3) \in E(G) \Rightarrow & (4i,4j + 2) \in E(G^{C_{NOT}}), (4i,4j + 3) \notin E(G^{C_{NOT}})\\
& w(4i,4j + 2)|_{G^{C_{NOT}}} = w(4i,4j + 3)|_{G},\\
(4i + 1,4j + 2) \in E(G) \Rightarrow & (4i + 1,4j + 3) \in E(G^{C_{NOT}}), (4i + 1,4j + 2) \notin E(G^{C_{NOT}})\\ & w(4i + 1,4j + 3)|_{G^{C_{NOT}}} = w(4i + 1,4j + 2)|_{G},\\
(4i + 1,4j + 3) \in E(G) \Rightarrow & (4i + 1,4j + 2) \in E(G^{C_{NOT}}), (4i + 1,4j + 3) \notin E(G^{C_{NOT}})\\ & w(4i + 1,4j + 2)|_{G^{C_{NOT}}} = w(4i + 1,4j + 3)|_{G}.
\end{split}
\end{equation}

\item
For the edges joining $C_{2i + 1} \subset \mathcal{C}_i$ and $C_{2j} \subset \mathcal{C}_j$, the following changes need to be made:
\begin{equation}
\begin{split}
(4i + 2,4j) \in E(G) \Rightarrow & (4i + 3,4j) \in E(G^{C_{NOT}}), (4i + 2,4j) \notin E(G^{C_{NOT}})\\
& w(4i + 3,4j)|_{G^{C_{NOT}}} = w(4i + 2,4j)|_{G},\\
(4i + 3,4j) \in E(G) \Rightarrow & (4i + 2,4j) \in E(G^{C_{NOT}}), (4i + 3,4j) \notin E(G^{C_{NOT}})\\
& w(4i + 2,4j)|_{G^{C_{NOT}}} = w(4i + 3,4j)|_{G},\\
(4i + 2,4j + 1) \in E(G) \Rightarrow & (4i + 3,4j + 1) \in E(G^{C_{NOT}}), (4i + 3,4j + 1) \notin E(G^{C_{NOT}})\\ & w(4i + 3,4j + 1)|_{G^{C_{NOT}}} = w(4i + 3,4j + 1)|_{G},\\
(4i + 3,4j + 1) \in E(G) \Rightarrow & (4i + 2,4j + 1) \in E(G^{C_{NOT}}), (4i + 3,4j + 1) \notin E(G^{C_{NOT}})\\ & w(4i + 2,4j + 1)|_{G^{C_{NOT}}} = w(4i + 3,4j + 1)|_{G}.
\end{split}
\end{equation}

\item
Change the edges joining modules $C_{2i + 1} \in \mathcal{C}_i$ and $C_{2j + 1} \in \mathcal{C}_j$.
\begin{equation}
\begin{split}
(4i + 2,4j + 4) \in E(G) \Rightarrow & (4i + 3,4j + 3) \in E(G^{C_{NOT}}), (4i + 2,4j + 4) \notin E(G^{C_{NOT}})\\& w(4i + 3,4j + 3)|_{G^{C_{NOT}}} = w(4i + 2,4j + 4)|_{G},\\
(4i + 3,4j + 3) \in E(G) \Rightarrow & (4i + 2,4j + 2) \in E(G^{C_{NOT}}), (4i + 3,4j + 3) \notin E(G^{C_{NOT}})\\& w(4i + 2,4j + 2)|_{G^{C_{NOT}}} = w(4i + 3,4j + 3)|_{G},\\
(4i + 2,4j + 3) \in E(G) \Rightarrow & (4i + 2,4j + 3) \in E(G^{C_{NOT}})\\
& w(4i + 2,4j + 3)|_{G^{C_{NOT}}} = w(4i + 2,4j + 3)|_{G},\\
(4i + 3,4j + 2) \in E(G) \Rightarrow & (4i + 3,4j + 2) \in E(G^{C_{NOT}})\\
& w(4i + 3,4j + 2)|_{G^{C_{NOT}}} = w(4i + 3,4j + 2)|_{G}.\\
\end{split}
\end{equation}
\end{enumerate}
\end{enumerate}
\end{pro}

\begin{thm}
$\rho(G^{C_{NOT}}) = {C_{NOT}}_n\rho(G){C_{NOT}}_n$.
\end{thm}

\begin{proof}
Proof follows from these block matrix multiplications.
\begin{equation*}
\begin{split}
{C_{NOT}}_n & = I \otimes I \otimes \dots \otimes C_{NOT} = \diag\{I, X, I, X \dots I, X\}\\
{C_{NOT}}_nA(G){C_{NOT}}_n & = \begin{bmatrix}
IC_0I & IC_{0,1}X & IC_{0,2}I & IC_{0,3}X & \dots \\
XC_{1,0}I & XC_1X & XC_{1,2}I & XC_{1,3}X & \dots \\
IC_{2,0}I & IC_{2,1}X & IC_2I & IC_{2,3}X & \dots \\
XC_{3,0}I & XC_{3,1}X & XC_{3,2}I & XC_3X & \dots \\
\vdots & \vdots & \vdots & \vdots & \vdots \\
\end{bmatrix}\\
& = \begin{bmatrix}
C_{NOT}\mathcal{C}_0C_{NOT} & C_{NOT}\mathcal{C}_{0,1}C_{NOT} & \dots \\
C_{NOT}\mathcal{C}_{1,0}C_{NOT} & C_{NOT}\mathcal{C}_{1}C_{NOT} & \dots \\
\vdots & \vdots & \dots \\
\end{bmatrix}\\
\end{split}
\end{equation*}
\end{proof}

We now use graph switching techniques to depict the action of Hadamard and CNOT gates to generate Bell states from two qubit separable states. The structure of Bell states was shown earlier in \cite{us}.

We begin with initial state $\ket{10}$. We operate a Hadamard gate on the first qubit followed by a CNOT gate to generate Bell state as follows,
$$\ket{10} \xrightarrow{H_1} \frac{1}{\sqrt{2}}(\ket{00} - \ket{10}) \xrightarrow{CNOT} \frac{1}{\sqrt{2}}(\ket{00} - \ket{11}).$$
Graph, corresponding to state $\ket{10}\bra{10}$, with vertex decomposition $\mathcal{C} = C_0 \cup C_1 = \{0, 1\} \cup \{2, 3\}$, is
$$
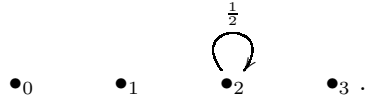
\xymatrix{{\bullet}_0 & {\bullet}_1 & {\bullet}_2 \ar@(ul,ur)[]^{\frac{1}{2}} & {\bullet}_3}.$$

To apply Hadamard gate on first qubit, i.e., $H_1$, we first swap vertices. The graph changes to
$$
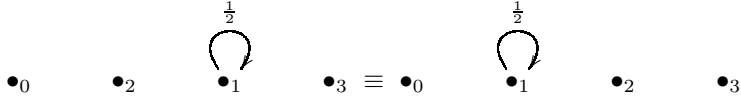
\xymatrix{{\bullet}_0 & {\bullet}_2 & {\bullet}_1 \ar@(ul,ur)[]^{\frac{1}{2}} & {\bullet}_3} \equiv \xymatrix{{\bullet}_0 & {\bullet}_1 \ar@(ul,ur)[]^{\frac{1}{2}} & {\bullet}_2 & {\bullet}_3}$$

Apply $H_2$ and get a new graph
$$\xymatrix{{\bullet}_0 \ar@{-}[r]^{-1}& {\bullet}_1 & {\bullet}_2 & {\bullet}_3}$$

To finish $H_1$ we swap it again. Graph after completing Hadamard operation is
$$
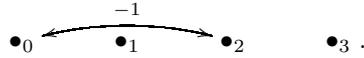
\xymatrix{{\bullet}_0 \ar@/^/[rr]^{-1} & {\bullet}_1 & {\bullet}_2 \ar@/_/[ll] & {\bullet}_3}.$$

Now apply CNOT operation. Following the procedure applied above, the new graph represents the state $\frac{1}{\sqrt{2}}(\ket{00} - \ket{11})$ \cite{us}.
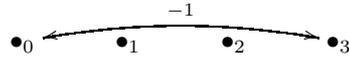
$$\xymatrix{{\bullet}_0 \ar@/^/[rrr]^{-1} & {\bullet}_1 & {\bullet}_2 & {\bullet}_3\ar@/_/[lll]}$$
Similarly all other Bell states can be generated graph theoretically.

\section{Conclusion}

In  this work  we  establish  a proof  of  principle for  representing
quantum   states  and   local  unitaries   graph  theoretically.    In
particular,  quantum states  are described  by the  signless Laplacian
matrix of their graph representation. We  work out in detail the graph
switching  operations  that  correspond  to some important local  unitaries  on  $n$
qubits.

While this has obvious significance in quantum information processing,
we     think    that     this     may have    impact on foundational
issues as well. Essentially, by representing  quantum superposition as a
weighted edge,  (to give  a dramatic  slant) our  approach geometrizes
non-realism, and herein lies its  appeal.  The full scope of this
approach   and  its   generalization  to   the  hypergraph   formalism
\cite{SPSS} will be discussed in future works.

 This  work is,  hopefully, a stepping  stone towards  a graph
theoretical understanding of issues in quantum foundations and quantum
information.

\end{document}